\documentclass[11pt]{article}

\usepackage[myheadings]{fullpage}

\usepackage[dvips,pdftex]{graphics}   
\usepackage{endnotes}            
\usepackage{amsfonts}            
\usepackage{amssymb}             
\usepackage{amsthm}              
\usepackage{amsmath}            
\usepackage[ruled,linesnumbered]{algorithm2e}           
\usepackage[letterpaper]{geometry}
\usepackage{rotating}            

\usepackage{setspace}                
\usepackage{xspace}
\newcommand{\ilogit}{\ensuremath{\mathrm{logit}^{-1}}\xspace}
\newcommand{\logit}{\ensuremath{\mathrm{logit}}\xspace}

\usepackage{natbib}
\begin{document}

\title{A Perfect Sampling Method for Exponential Family Random Graph Models\thanks{This research was supported by ONR award \#N00014-08-1-1015 and NSF award IIS-1526736.} \thanks{The author would like to thank Johan Koskinen, Mark Handcock, David Hunter, and John Skvoretz for their helpful input.}
}

\author{
Carter T. Butts\thanks{Department of Sociology, Statistics, and EECS and Institute for Mathematical Behavioral Sciences; University of California, Irvine; SSPA 2145; Irvine, CA 92697; \texttt{buttsc@uci.edu}}
}



\date{10/11/17; to appear in the \emph{Journal of Mathematical Sociology}
}
\maketitle



\begin{abstract}
Generation of deviates from random graph models with non-trivial edge dependence is an increasingly important problem.  Here, we introduce a method which allows perfect sampling from random graph models in exponential family form (``exponential family random graph'' models), using a variant of Coupling From The Past.  We illustrate the use of the method via an application to the Markov graphs, a family that has been the subject of considerable research.  We also show how the method can be applied to a variant of the biased net models, which are not exponentially parameterized.\\[5pt]
\end{abstract}

\newtheorem{theorem}{Theorem}




\section{Introduction}

Simulation of random graph processes is an increasingly important problem in many fields.  This is particularly true in the social and biological sciences, where graphs are used to represent such diverse phenomena as interpersonal communication, collaboration among organizations, trophic systems, and protein-protein interaction networks.  Networks encountered in such fields typically exhibit patterns of complex dependence, in the sense that the state of one edge frequently depends on the state of other edges in the network (even when other aspects of structure are taken into account).  Parameterization of models for such networks is a difficult problem, and has spawned a range of approaches (see, e.g., \citet{watts.strogatz:n:1998,barabasi.albert:s:1999,hoff.et.al:jasa:2002,newman:siamr:2003,skvoretz.et.al:sn:2004,butts:jms:2015}).  A particularly significant trend in recent years has been the use of discrete exponential families for the parameterization of networks with complex dependence, following the important early work of \citet{holland.leinhardt:jasa:1981} \citet{holland.et.al:sn:1983}, \citet{frank.strauss:jasa:1986}, and \citet{wasserman.pattison:p:1996}.  As with parallel developments in the spatial statistics literature \citep{besag:ts:1975,ripley:jrssB:1977,strauss:siam:1986}, discrete exponential families have provided a ``lingua franca'' for the description of random graph models, along with a fairly well-developed body of inferential and computational theory.  These attractive features have led to a significant expansion in the use of exponentially parameterized random graph models (frequently called ``exponential family random graph'' or ERG models) within the scientific literature.

Despite their obvious utility, ERG models pose some pragmatic challenges.  In particular, few  properties of most non-trivial ERG models are susceptible to analytical treatment, and simulation is thus required to study ERG behavior.  This is true for both deductive (i.e., discovering model properties) and inferential (i.e., estimating model parameters from data) applications.  Direct simulation of ERG models is generally infeasible due to the presence of an unknown normalizing factor, which involves summation of an extremely rough function across the (very large) support.  The practical solution to this problem has been the use of Markov chain Monte Carlo (MCMC) methods, which allow for approximate simulation from the target distribution without the need to compute the normalizing factor.  Unfortunately, convergence of such procedures is non-trivial to assess, and may be poor when models are near the ``degenerate'' regions of their parameter spaces \citep{strauss:siam:1986,snijders:joss:2002,handcock:ch:2003,bhamidi.et.al:aap:2011}.  Even where degeneracy is not a concern, MCMC is ill-suited to generating samples of provably high quality for use in algorithm evaluation, method testing, or high-precision applications.  Here, we propose to address this problem via a perfect sampling method, based on the Coupling From The Past (CFTP) technique of \citet{propp.wilson:rsa:1996}.  This method can be used with any ERG model, but is particularly well-suited to ERGs whose statistics take the form of subgraph counts.  Such statistics arise naturally via the Hammersley-Clifford Theorem \citep{besag:jrssB:1974} when ERGs are parameterized using dependency hypotheses (see, e.g., \citet{pattison.robins:sm:2002,wasserman.robins:ch:2005}) and are the basis for important ERG families such as the Markov graphs \citep{frank.strauss:jasa:1986}.  The method can also be used with certain random graph families which are not parameterized in ERG form, but which can be specified via their edgewise full conditionals; we discuss this in Section~\ref{sec_bn} in the context of the ``biased net'' models of \citet{rapoport:bmb:1949a,rapoport:bmb:1949b,rapoport:bmb:1950}.

The remainder of the paper proceeds as follows.  We begin with general concepts and notation, including a brief overview of ERG models.  We then present our simulation method, along with techniques for more efficient computation on models based on subgraph census statistics.  In the following section, we apply our simulation method to the Markov graphs, illustrating its use with simulations from the edge clustering and triangle models.  Finally, we close with a brief discussion of extensions and generalizations, including the generation of deviates from locally-parameterized biased net models.

\subsection{Notation and Core Concepts}

For the most part, we will focus here on simple graphs of finite order.  These may be represented by ordered pairs $G=(V,E)$, where $V$ is a set of \emph{vertices} and $E$ is a set of \emph{edges} on $V$.  For a simple graph, the elements of $E$ are two-element subsets of $V$.  Another important class of objects is the class of directed graphs (or digraphs), for which $E$ is a subset of ordered pairs on $V$.  When working with a fixed vertex set, we will let $n=|V|$ be the \emph{order} of $G$ (with $|\cdot|$ denoting cardinality).  In practice, it is usually convenient to represent graphs via their \emph{adjacency matrices}; the adjacency matrix, $y$, of graph $G$ is the $n \times n$ binary matrix such that $y_{ij}=1$ if $\{i,j\} \in E$ and $y_{ij}=0$ otherwise.  For simple graphs, it is clearly the case that $y_{ij}=y_{ji}$ and $y_{ii}=0$.  The latter constraint is preserved for simple digraphs, but not the former.  We will frequently need to refer to random graphs (undirected or directed), i.e., random variables whose sample space consists of a graph set.  We describe these via their (random) adjacency matrices, using capital letters -- thus, if $Y$ is the adjacency matrix of an undirected random graph, $Y_{ij}$ is the random variable indicating the presence or absence of an $\{i,j\}$ edge.  Likewise, we can describe a stochastic process on a set of graphs (a random graph process) by a sequence of random adjacency matrices $Y^{(1)},Y^{(2)},\ldots$.  Throughout this text, we will use parenthetical superscripts to index both sequences of variables and their realizations within a random process; thus, a realization of a random graph process $Y^{(1)},Y^{(2)},\ldots$ would be denoted $y^{(1)},y^{(2)},\ldots$.

Given two graphs $G$ and $H$, we say that $G$ is a \emph{subgraph} of $H$ if $V(G) \subseteq V(H)$ and $E(G) \subseteq E(H)$.  Clearly, if $y$ and $y'$ are the adjacency matrices of $G$ and $H$ (respectively), $G \subseteq H$ then implies that $y_{ij} \le y'_{ij}$ for all $i,j$; we thus denote the latter relationship by the $\subseteq$ operator as well, where there is no danger of confusion.  Let $K_n$ and $N_n$ denote the \emph{complete} and \emph{empty} graphs of order $n$ (i.e., the order-$n$ graphs having respectively all or no edges).  Then $\subseteq$ forms a partial order on the set of order-$n$ graphs, with unique upper bound $K_n$ and unique lower bound $N_n$.  We shall make use of this observation in the presentation which follows.  For expository purposes, it is also convenient to introduce a simplified notation for graphs which are perturbed by forcing a given edge to be present or absent, and for the edge variables of a graph excluding a particular element.  We do this via adjacency matrices.  Let $Y^c_{ij}$ refer to the set of all edge variables other than $\{i,j\}$ (or $(i,j)$ in the directed case) in random adjacency matrix $Y$; the corresponding observations are denoted $y^c_{ij}$.  $Y^+_{ij}$ is then defined as the random matrix with $\left(Y^+_{ij}\right)_{kl}=Y_{kl}$ for $\{i,j\} \neq \{k,l\}$ (directed case: $(i,j) \neq (k,l)$) and $Y_{ij}=1$.  $Y^-_{ij}$ is similarly defined as the random matrix such that $\left(Y^-_{ij}\right)_{kl}=Y_{kl}$ for $\{i,j\} \neq \{k,l\}$ (respectively, $(i,j) \neq (k,l)$) and $Y^-_{ij}=0$.  We apply this notation to realizations as well, i.e. $y^+_{ij}$ and $y^-_{ij}$ are equal to $y^c_{ij}$ with $y^+_{ij}=1$ and $y^-_{ij}=0$, and to matrix sets, i.e. $\mathcal{A}^+_{ij}=\left\{y^+_{ij}; y \in \mathcal{A}\right\}$ and $\mathcal{A}^-_{ij}=\left\{y^-_{ij}; y \in \mathcal{A}\right\}$.

Our principal concern within the paper will be the simulation of draws from exponentially parameterized random graph distributions on graphs of fixed order.  Let $Y$ be an order-$n$ adjacency matrix, and let $\mathcal{Y}_n$ be the set of such matrices.  Then we may write the pmf of $Y$ in exponential family form\footnote{For simplicity, we take $Y$ to be parameterized with respect to the counting measure on $\mathcal{Y}_n$.  Where other reference measures are desired (e.g. Krivitsky's (\citeyear{krivitsky.et.al:statm:2011}) constant mean degree reference), this can be accomplished by folding them into $t$.} as
\begin{equation}
\Pr\left(Y=y\left|t,\theta\right.\right) = \frac{\exp\left(\theta^T t\left(y\right)\right)}{\sum_{y' \in \mathcal{Y}_n} \exp\left(\theta^T t\left(y'\right)\right)} I_{\mathcal{Y}_n}(y), \label{eq_erg}
\end{equation}
where $t: \mathcal{Y}_n \mapsto \mathbb{R}^p$ is a vector of statistics, $\theta \in \mathbb{R}^p$ is a vector of parameters, and $I_{\mathcal{Y}_n}$ is an indicator function for membership in $\mathcal{Y}_n$.  In general, computation involving this pmf is complicated by the practical impossibility of directly computing the normalizing factor, $\sum_{y' \in \mathcal{Y}_n} \exp\left(\theta^T t\left(y'\right)\right)$: since $|\mathcal{Y}_n|$ is of order $2^{n^2}$, explicit summation is prohibitive for all but the smallest graphs.  Moreover, the considerable roughness of $\exp\left(\theta^T t\left(y\right)\right)$ over the support of $Y$ makes simple Monte Carlo quadrature schemes ineffective.  Typically, simulation schemes exploit the fact that the normalizing factor is not needed to compute probability ratios given fixed $\theta$, i.e.,
\begin{equation}
\frac{\Pr\left(Y=y'\left|t,\theta\right.\right)}{\Pr\left(Y=y\left|t,\theta\right.\right)} = \exp\left(\theta^T\left(t\left(y'\right)-t\left(y\right)\right)\right) \label{eq_prat}
\end{equation}
for $y,y' \in \mathcal{Y}_n$.  This lends itself neatly to Markov chain Monte Carlo algorithms (such as Gibbs or Metropolis-Hastings samplers) which require that the target distribution be specified only up to a normalizing constant.  While useful in many settings, MCMC methods have the well-known disadvantage of being approximate sampling algorithms whose adequacy can be difficult to verify (see, e.g., \citet{gamerman:bk:1997,gelman:ch:1996}).  This is of particular concern in settings such as likelihood approximation \citep{geyer.thompson:jrssC:1992,hunter.handcock:jcgs:2006}, wherein poor-quality MCMC samples may in turn adversely affect estimation.  Since many intuitively attractive ERG models are degenerate or near-degenerate for large portions of their parameter spaces \citep{strauss:siam:1986,handcock:ch:2003,schweinberger:jasa:2011}, this is a potentially serious problem; indeed, the challenges of simulation and inference in near-degenerate settings have been a major concern of those implementing tools for practical use \citep{hunter.et.al:jss:2008}.  While advances in ERG parameterization have greatly extended the range of families for which degeneracy is less of a concern \citep{lusher.et.al:bk:2012}, simulation quality is still potentially important for applications such as high-precision likelihood calculations, generation of high-quality samples against which to check approximate simulation or inference methods \citep[e.g.][]{pu.et.al:nips:2012,butts:jms:2015}.

\section{Simulation Method}

As indicated above, our focus here is on the development of a general method for perfect (sometimes called ``exact'') sampling from fixed-order ERG distributions.  Our approach falls within the general family of methods known as ``Coupling From The Past'' \citep{propp.wilson:rsa:1996}, so-called because it involves the use of coupled Markov chains extended backwards through (virtual) time.  The base chain employed for this purpose is the well-known single (edge) update Gibbs sampler, a frequently used tool for approximate simulation of ERG models.  Although the base chain is non-monotone, coalescence detection is made possible by constructing a two chain bounding process whose elements ``sandwich'' the states of the base chain. The bounding approach employed here has previously been exploited for non-MCMC based approximate ERG sampling \citep{butts:jms:2015} and for the derivation of analytical bounds on ERG behavior \citep{butts:sm:2011b}.  In that it makes use of bounding processes that differ from the target Markov chain, our simulation method has some resemblance to dominated CFTP \citep{kendall:proc:1997,kendall.moller:aap:2000} (also called ``Coupling Into and From The Past''); since we employ the bounding processes only for coalescence detection, however, and not for coupling, our approach is actually more similar to ``classic'' CFTP than to dominated CFTP.  

Our presentation of the simulation method begins by reviewing the single-update Gibbs sampler for ERGs.  We then discuss the bounding processes employed to ``sandwich'' the states of the sampler, including the computation of change score bounds to facilitate implementation and the use of the bounding processes in coalescence detection.  This is followed by the presentation of a unified algorithm for the perfect sampling scheme.  

\subsection{Underlying Gibbs Sampler}

Our simulation method is built on a familiar sampling procedure for ERG families, the single-update Gibbs sampler (see, e.g., \citet{snijders:joss:2002}).  This procedure may be described as follows.  Define $\Delta_{ij}(y)=t\left(y^+_{ij}\right)-t\left(y^-_{ij}\right)$ to be the vector of ``change scores'' for $t$ on adjacency matrix $y$, given a perturbation of the $i,j$ edge.  We note that, for an ERG family with sufficient statistics $t$ and parameter vector $\theta$, $Y_{ij}$ is conditionally Bernoulli distributed with parameter
\begin{align}
\Pr\left(Y_{ij}=1 \left| y^c_{ij}, t, \theta\right.\right) &= \frac{1}{1+\exp\left(-\theta^T \Delta_{ij}\left(y\right)\right)}\\  \label{eq_gibbs_update}
&=\ilogit\left(\theta^T \Delta_{ij}\left(y\right)\right)
\end{align}
This is a direct consequence of Equation~\ref{eq_prat}.  Now, consider a sequence of matrices $Y^{(1)},Y^{(2)},\ldots$ formed by identifying a vertex pair $\{i,j\}$ (directed case: $(i,j)$) at each step, and letting $Y^{(i)}=\left(Y^{(i-1)}\right)^+_{ij}$ with probability given by Equation~\ref{eq_gibbs_update} and $Y^{(i)}=\left(Y^{(i-1)}\right)^-_{ij}$ otherwise.  Subject to fairly mild conditions on the choice of $\{i,j\}$ (e.g., all pairs chosen with positive probability within some bounded number of steps, and choice of pair independent of $Y$) and the finiteness of $\theta^T t(Y)$, $Y^{(1)},Y^{(2)},\ldots$ forms a Markov chain with equilibrium distribution given by Equation~\ref{eq_erg}.  Although convergence of this procedure may be slow (see \citet{snijders:joss:2002} of \citet{bhamidi.et.al:aap:2011} for a discussion), it is easily implemented and enjoys the substantial benefit that computation of the change scores (i.e., $\Delta$) can often be performed in constant or linear time.  For our purposes, this scheme is also useful because it allows for the specification of \emph{bounding processes} which allow for coalescence detection in the context of a CFTP algorithm.  It is to the definition of these processes that we now turn.

\subsection{Definition of the Bounding Processes}

Given a single-update Gibbs sampler as defined above, it is possible to define a pair of graph processes which stochastically bound the former process in terms of the subgraph relation.  Let $(L,U)$ be the ``lower'' and ``upper'' processes, respectively; our aim is to construct these processes in such a way as to ensure that $L^{(i)} \subseteq Y^{(i)} \subseteq U^{(i)}$ for all $i\ge 0$ and for all realizations of $Y$.  Let us assume that, for some given $i$, the condition $L^{(i)} \subseteq Y^{(i)} \subseteq U^{(i)}$ holds, and let $\mathcal{B}^{(i)}=\{y \in \mathcal{Y}_n: L^{(i)} \subseteq y \subseteq U^{(i)}\}$ be the set of adjacency matrices bounded by the upper and lower processes at iteration $i$.  The evolution of $(L,U)$ is governed by two vectors of change score functions, $\Delta^L$ and $\Delta^U$, with elements constructed from $\Delta$ for a given graph set $\mathcal{A}$ as follows:
\begin{gather}
\Delta^L_{ij}\left(\mathcal{A},\theta\right)_k = \begin{cases} \max_{y \in \mathcal{A}} \Delta_{ij}(y)_k& \theta_k \le 0\\ \min_{y \in \mathcal{A}} \Delta_{ij}(y)_k& \theta_k > 0\end{cases} \label{eq_deltal}\\
\Delta^U_{ij}\left(\mathcal{A},\theta\right)_k = \begin{cases} \min_{y \in \mathcal{A}} \Delta_{ij}(y)_k& \theta_k \le 0\\ \max_{y \in \mathcal{A}} \Delta_{ij}(y)_k& \theta_k > 0\end{cases}. \label{eq_deltau}
\end{gather}

As with the single-update Gibbs sampler, we assume that at the $i$th iteration some pair $j,k$ has been chosen for updating; further, we assume that we are given a sequence $u^{(0)}, u^{(1)}, \ldots$ of iid uniform random deviates on the $[0,1]$ interval.  The bounding processes then simultaneously evolve by the following updating mechanism:
\begin{gather}
L^{(i+1)} = \begin{cases}\left(L^{(i)}\right)^+_{jk} & u^{(i)} \le \ilogit\left(\theta^T \Delta^L_{jk}\left(\mathcal{B}^{(i)},\theta\right) \right) \\ \left(L^{(i)}\right)^-_{jk} & u^{(i)} > \ilogit\left(\theta^T \Delta^L_{jk}\left(\mathcal{B}^{(i)},\theta\right) \right) \end{cases}\\
U^{(i+1)} = \begin{cases}\left(U^{(i)}\right)^+_{jk} & u^{(i)} \le \ilogit\left(\theta^T \Delta^U_{jk}\left(\mathcal{B}^{(i)},\theta\right) \right)\\ \left(U^{(i)}\right)^-_{jk} & u^{(i)} > \ilogit\left(\theta^T \Delta^U_{jk}\left(\mathcal{B}^{(i)},\theta\right) \right)\end{cases}. \label{eq_upprobu}
\end{gather}
Observe that, under this updating rule, the probability of setting $U^{(i+1)}_{jk}=1$ is greater than or equal to the probability of setting $Y^{(i+1)}_{jk}=1$ (since $\Delta^U$ is constructed so as to strictly favor edge addition).  Thus, if $Y^{(i)}\subseteq U^{(i)}$, then $Y^{(i+1)}\subseteq U^{(i+1)}$ (assuming that all states are updated using the same ``random coins,'' $u$).  Likewise, the probability of setting $L^{(i+1)}_{jk}=0$ is greater than or equal to the probability of setting $Y^{(i+1)}_{jk}=0$, and thus if $L^{(i)}\subseteq Y^{(i)}$, then $L^{(i+1)}\subseteq Y^{(i+1)}$.  We can guarantee that the initial condition holds for both chains by setting $L^{(0)}=N_n$ and $U^{(0)}=K_n$ (the lower and upper bounds on $\mathcal{Y}_n$, respectively).  By induction, it then follows that $L^{(i)}\subseteq Y^{(i)}\subseteq U^{(i)}$ for all $i>0$, and all $Y$.

\subsubsection{Bounding the Change Scores}

In the above construction, calculation of $\Delta^L$ and $\Delta^U$ is obviously an important consideration: if we must examine every element of $\mathcal{B}$ for this purpose, then simulation of the bounding processes will be impractical. (Recall that, in the initial condition, the set of bounded graphs is equal to $\mathcal{Y}_n$.)  Thankfully, such enumeration is typically unnecessary.  In particular, let us assume that $t$ is such that $t_i\left(Y\right)\le t_i\left(Y'\right)$ for all $Y \subseteq Y'$ (i.e., the elements of $t$ are weakly monotone increasing in edge addition).  In this case, $\Delta_{jk}$ clearly cannot be greater than the difference between $t$ evaluated on $U^+_{jk}$ and $t$ evaluated on $L^-_{jk}$; since edge addition can only increase $t$, it also follows that $\Delta_{jk}$ is nonnegative.  It follows therefore that $\max_{y \in \mathcal{B}^{(i)}} \Delta_{jk}(y) \le t\left(U^+_{jk}\right)-t\left(L^-_{jk}\right)$, and $\min_{y \in \mathcal{B}^{(i)}} \Delta_{jk}(y) \ge 0$.  Substituting these bounds for those used in Equations~\ref{eq_deltal}--\ref{eq_deltau} preserves the dominance properties of the bounding processes, and requires only the evaluation on change scores on two graphs (as opposed to the entire bounded set).

Further refinement is possible when $t$ is such that $\Delta$ itself is at least weakly monotone increasing in edge addition.  In this case, it is trivially true that $\max_{y \in \mathcal{B}^{(i)}} \Delta_{jk}(y) \le t\left(U^+_{jk}\right)-t\left(U^-_{jk}\right)$ and $\min_{y \in \mathcal{B}^{(i)}} \Delta_{jk}(y) \ge t\left(L^+_{jk}\right)-t\left(L^-_{jk}\right)$.  Since at least one member of $\mathcal{B}^{(i)}$ exhibits each of these values ($U$ and $L$, specifically), these bounds are the tightest possible.  As in the above case, substituting these bounds in the definition of $\Delta^L$ and $\Delta^U$ allows for the $L$ and $U$ to be updated without the necessity of calculating $t$ for all members of $\mathcal{B}^{(i)}$.

It should be noted that the latter case is of particular interest, since it encompasses all \emph{subgraph census statistics} (i.e., statistics which consist of the number of copies of a given isomorphism class within $y$).  This can be understood as follows: let $H$ be an isomorphism class (to be counted), and let $\mathcal{H}_{ij}$ be the set of ``edge-missing preconditions'' for $H$ -- that is, the set of subgraphs $H'$ which are isomorphic to $H$ given the addition of the $i,j$ edge.  Let $t$ be a subgraph census statistic counting copies of $H$.  Then $\Delta_{ij}(y)$ is clearly equal to the number of copies of all $H' \in \mathcal{H}_{ij}$ belonging to $y^-_{ij}$.  Since adding non-$ij$ edges to $y$ cannot decrease the number of pre-condition subgraphs, it follows that $\Delta_{ij}(y)\le\Delta_{ij}(y')$ for all $y \subseteq y'$.  Subgraph census statistics arise naturally from the Hammersley-Clifford Theorem \citep{besag:jrssB:1974} when homogeneity constraints are applied to statistics indicating members of the same isomorphism class; they are used extensively in the modeling of social networks (see, e.g., \citet{holland.leinhardt:jasa:1981,frank.strauss:jasa:1986,wasserman.pattison:p:1996,pattison.wasserman:bjmsp:1999,pattison.robins:sm:2002}).  We shall consider a specific example in Section~\ref{sec_mgbounds}, when we apply these results to the case of the Markov graphs.

\subsection{Coalescence Detection}

Let $\ldots,Y^{(-1)},Y^{(0)},Y^{(1)},\ldots$ be the states of a Markov chain resulting from a single-update Gibbs sampler as described above, and for some $i>0$ let $L^{(-i)}=N_n, U^{(-i)}=K_n$ (where $(L,U)$ are the bounding processes associated with $Y$).  Suppose that, in the joint evolution of $(L,Y,U)$, there exists some time $-j$ such that $-i\le -j \le 0$ and $L^{(-j)}=U^{(-j)}$.  By construction, $L^{(-j)} \subseteq Y^{(-j)} \subseteq U^{(-j)}$, and hence $Y^{(-j)}=L^{(-j)}=U^{(-j)}$.  Moreover, since $L^{(-i)} \subseteq Y^{(-i)} \subseteq U^{(-i)}$ for all possible $Y^{(-i)}$, it follows that all past sequences $\ldots,Y^{(-i-1)},Y^{-i}$ lead to $Y^{(-j)}$; by extension, $Y^{(0)}$ must be a draw from the infinite history of $Y$.\footnote{Note that we cannot simply take $Y^{(-j)}$, since the coalescence point is not an independent draw from the equilibrium distribution of $Y$.  Fixing the sampling time in advance resolves this difficulty.}  $Y$, however, is by construction a Gibbs sampler with unique equilibrium distribution corresponding to Equation~\ref{eq_erg}.  \emph{Thus, $Y^{(0)}$ is distributed as an ERG with statistics $t$ and parameter vector $\theta$.}

This phenomenon -- by which all trajectories of a Markov chain beyond a given point converge to a single state -- is known as \emph{coalescence} \citep{propp.wilson:rsa:1996}.  Our use of the bounding chains, then, is a \emph{coalescence detection} scheme for $Y$; observing the event $L^{(-j)}=U^{(-j)}$ tells us that $Y$ has coalesced, without requiring explicit computation of all possible chains from $Y^{(-i)}$ to $Y^{(0)}$.  Of course, there is no guarantee that, for a given $i$ (and associated sequence of updates), $Y$ will have coalesced by time 0.  In this case, however, one can recede further into the past, and try again.  Once one finds a case for which coalescence has been detected, one can take the resulting value of $Y^{(0)}$ as a draw from the target distribution.

\subsubsection{Time to Coalescence}

Per the above, if $U$ and $L$ ever coincide, then coalescence has occurred.  What can be said regarding the time to coalescence?  The key result is expressed in the following theorem.

\begin{theorem} \label{thm_coal}
Let $L, U$ be bounding processes for $Y$ having finite $\theta^T t(y)$ on $y\in \mathcal{Y}_n$, with an underlying Gibbs sampler that updates all edge variables within every $s$ iterations.  Let $C_i$ be an indicator for the event that $L^{(j)}=U^{(j)}$ for some $j\le i$.  Then (i) $\Pr(C_i=1)\to 1$ as $i \to \infty$ and (ii) there exists some $\epsilon>0$ such that $\Pr(C_{ks}=1)\ge 1-(1-\epsilon^s)^k$ for $k\in 1,2,\ldots$.
\end{theorem}
\begin{proof}
Our proof proceeds as follows.  First, we observe that each edge update has a positive probability of setting an element of $L$ and $U$ equal to each other.  We then note that there exists a positive probability that a sequence $s$ of such updates will set all elements of $L$ equal to all elements of $U$.  The coalescence time results then follow.

To begin, assume we are at the $i$th iteration of the process, with the pair $j,k$ selected for updating.  From Equations~\ref{eq_deltal}-\ref{eq_upprobu}, it follows that the probability of the transition pair $L^{(i+1)}=\left(L^{(i)}\right)^+_{jk}, U^{(i+1)}=\left(U^{(i)}\right)^+_{jk}$ is given by $\ilogit\left(\theta^T \Delta^L_{jk}\left(\mathcal{B}^{(i)},\theta\right) \right) \ge \min_{jk} \ilogit\left(\theta^T \Delta^L_{jk}\left(\mathcal{Y}_n,\theta\right) \right)$.  Likewise, it also follows that the probability of the transition pair $L^{(i+1)}=\left(L^{(i)}\right)^-_{jk}, U^{(i+1)}=\left(U^{(i)}\right)^-_{jk}$ is given by $\ilogit\left(-\theta^T \Delta^U_{jk}\left(\mathcal{B}^{(i)},\theta\right) \right) \ge \min_{jk} \ilogit\left(-\theta^T \Delta^U_{jk}\left(\mathcal{Y}n,\theta\right) \right)$.  By the finiteness of $\theta$ and $t(y)$, it follows that there exists some $\epsilon>0$ such that $\min_{jk} \ilogit\left(\theta^T \Delta^L_{jk}\left(\mathcal{Y}_n,\theta\right) \right) > \epsilon$ and $\min_{jk} \ilogit\left(-\theta^T \Delta^U_{jk}\left(\mathcal{Y}n,\theta\right) \right) > \epsilon$.  $\epsilon$ is then a lower bound on the probability that $\left(L^{(i+1)}\right)_{jk}=\left(U^{(i+1)}\right)_{jk}$.

Next, assume once more that we are at the $i$th iteration of the process, and consider the next $s$ iterations.  From the above, each iteration sets the updated edge variable in $L$ equal to the corresponding variable in $U$ with probability $\ge \epsilon$.  By assumption, the sampler visits every edge variable within every $s$ iterations. Since obtaining an ``equalizing'' step at each such iteration will necessarily set all elements of $L$ equal to those of $U$, it follows that $\Pr(L^{(i+s)}=U^{(i+s)})\ge \epsilon^s$.

We now observe that, since the above results hold irrespective of the states of $L$ and $U$, the probability that coalescence will occur within $k$ blocks of $s$ iterations is greater than or equal to $1-(1-\epsilon^s)^k$.  Since $\epsilon>0$, it follows that the probability of coalescence approaches 1 as $k\to \infty$. 
\end{proof}

The essential intuition of Theorem~\ref{thm_coal} is that there is always some positive probability that $U$ and $L$ will draw closer to each other, and hence they eventually coincide with probability 1.  Moreover, the probability of coalescence increases exponentially fast in the number of iterations.  While this result does not guarantee that the expected time to coalescence will be small in practical terms (this depends on the transition probabilities), it does guarantee the existence of a time scale on which coalescence can be made arbitrarily likely.  It should be noted that where $\epsilon$ and $s$ satisfying the required conditions can be determined for a particular model, the second statement of the theorem allows for stronger statements to be made (e.g., upper bounds on the median or expected coalescence time). 

As a final note, it should be observed that while the coincidence of $U$ and $L$ is a sufficient condition for the coalescence of $Y$, it may not be a \emph{necessary} condition; if not, there may be other procedures (perhaps more efficient) that can also detect coalescence.  Although this question is not pursued further here, it is an intriguing possibility for future work in this area.

\subsection{Algorithm}

Putting all this together, Algorithm~\ref{alg_erg} shows a sample procedure for the use of exact sampling to generate ERG draws.  The approach taken is typical of CFTP algorithms (see, e.g. \citep{propp.wilson:rsa:1996}), combining forward evolution of the Markov chains with a geometric ``backing off'' procedure where coalescence is not obtained.  The initial chain depth is set to $\tbinom{n}{2}$ (line~\ref{alg_line_initi}), since at least this many updates are required for coalescence detection.  The random inputs to this algorithm are the ``coins,'' $u$, and the edges to update (stored as row/column pairs $r,c$); these are initialized in lines~\ref{alg_line_initrstart}--\ref{alg_line_initrstop}, via uniform draws from the appropriate distributions.  The main loop of the procedure (lines~\ref{alg_line_mainstart}--\ref{alg_line_mainstop}) initializes the bounding chains, runs them forward in time, and (if coalescence is not detected by time 0) backs off by a factor of two.  Once coalescence is detected (line~\ref{alg_line_coaldet}), $Y$ is set equal to the current bounding chain state and further updates are made to $Y$ rather than $L$ and $U$ (lines~\ref{alg_line_yupdstart}--\ref{alg_line_yupdstop}).  The value of the coalesced $Y$ at time 0 is then returned.

\dontprintsemicolon
\begin{algorithm} \scriptsize
\caption{Exact Sampling Procedure for Undirected ERGs \label{alg_erg}}
\KwData{$t$, $\theta$, $n$}
\KwResult{A single draw from $\mathrm{ERG}(t,\theta)$}
Let \textsc{Coalesced}$:=$\textsc{False}\;
Let $i:=\tbinom{n}{2}$\; \label{alg_line_initi}
Draw $u^{(-i)},\ldots,u^{(-1)} \sim \mathrm{Unif}(0,1)$\; \label{alg_line_initrstart}
Draw $r^{(-i)},\ldots,r^{(-1)}\sim \mathrm{Unif}(\{1,\ldots,n\})$\;
Draw $c^{(-i)},\ldots,c^{(-1)}\sim \mathrm{Unif}(\{1,\ldots,r-1\})$\; \label{alg_line_initrstop} 
\While{$\neg$\textsc{Coalesced}}{ \label{alg_line_mainstart}
  Let $L^{(-i)}:=N_n$\;
  Let $U^{(-i)}:=K_n$\;
  \For(\emph{(Evolve chains forward in time)}){$j\in -i,\ldots,-1$}{
    \eIf{$\neg$\textsc{Coalesced}}{
      \eIf(\emph{(Update $L$)}){$u^{(j)}\le \logit\left(\theta^T \Delta^L_{r^{(j)}c^{(j)}}\left(\mathcal{B}^{(j)},\theta\right)\right)$}{
        Let $L^{(j+1)}:=\left(L^{(j)}\right)^+_{r^{(j)}c^{(j)}}$\;
      }{
        Let $L^{(j+1)}:=\left(L^{(j)}\right)^-_{r^{(j)}c^{(j)}}$\;
      }
      \eIf(\emph{(Update $U$)}){$u^{(j)}\le \logit\left(\theta^T \Delta^U_{r^{(j)}c^{(j)}}\left(\mathcal{B}^{(j)},\theta\right)\right)$}{
        Let $U^{(j+1)}:=\left(U^{(j)}\right)^+_{r^{(j)}c^{(j)}}$\;
      }{
        Let $U^{(j+1)}:=\left(U^{(j)}\right)^-_{r^{(j)}c^{(j)}}$\;
      }
      \If(\emph{(Check for coalescence)}){$L^{(j+1)}=U^{(j+1)}$}{ \label{alg_line_coaldet}
        Let $Y^{(j+1)}:=L^{(j+1)}$\;
        Let \textsc{Coalesced}$:=$\textsc{True}\;
      }
    }{
      \eIf(\emph{(Update $Y$)}){$u^{(j)}\le \logit\left(\theta^T \Delta_{r^{(j)}c^{(j)}}\left(Y^{(j)}\right)\right)$}{ \label{alg_line_yupdstart}
        Let $Y^{(j+1)}:=\left(Y^{(j)}\right)^+_{r^{(j)}c^{(j)}}$\;
      }{
        Let $Y^{(j+1)}:=\left(Y^{(j)}\right)^-_{r^{(j)}c^{(j)}}$\;
      } \label{alg_line_yupdstop}
    }  
  }
  \If(\emph{(Recede farther into the past, if needed)}){$\neg$\textsc{Coalesced}}{
    \For{$j \in 1,\ldots,i$}{
      Draw $u^{(-i-j)} \sim \mathrm{Unif}(0,1)$\;
      Draw $r^{(-i-j)} \sim \mathrm{Unif}(\{1,\ldots,n\})$\;
      Draw $c^{(-i-j)} \sim \mathrm{Unif}(\{1,\ldots,r-1\})$\;
    }
    Let $i:=2i$\;
  }
} \label{alg_line_mainstop}
\Return{$Y^{(0)}$}
\end{algorithm}

Although exact running time will obviously vary with the implementation of $\Delta$ (and the mixing properties of the underlying chain), the need for $L$ and $U$ to meet ensures that Algorithm~\ref{alg_erg} is at least order $n^2$.  As such, it may be impractical for extremely large networks (e.g., those with tens of thousands of nodes).  On the other hand, such scaling is not prohibitive for networks of the size typically encountered in organizational or group settings.  Due to the geometric backing-off procedure, the coalescence point is guaranteed to be found in $\log_2 T$ iterations of the main loop, where $-T$ is the coalescence time; similarly, no more than $T$ ``excess'' updating steps are employed in the final iteration.  Efficient implementation of Algorithm~\ref{alg_erg} depends critically on the change score computation, which should in practice be optimized to the extent feasible.  For a discussion of this and related issues, see \citet{hunter.et.al:jss:2008}.

\section{Application to the Markov Graphs}

First introduced by \citet{frank.strauss:jasa:1986}, the Markov graphs constitute one of the most basic ERG families with complex dependence.  Specifically, the Markov graphs are formed by the set of distributions for which the states of two potential edges, $\{i,j\},\{k,l\}$, are conditionally dependent only if they have an endpoint in common (i.e., $\{i,j\}\cap\{k,l\} \neq \varnothing$).  This condition can be a viewed as a graph-theoretic version of parallel developments in spatial statistics (e.g., \citet{besag:ts:1975}), in which states associated with particular locations are conditionally dependent only if those locations share a border (or other equivalent notion of contact).  As Frank and Strauss demonstrate, the sufficient statistics for the (homogeneous) Markov graphs in the undirected case are the counts of $k$-stars (i.e., copies of $K_{1,k}$) and triangles (i.e., copies of $K_3$).  Since the $k$-star census has a one-to-one relationship with the degree distribution, the undirected Markov graphs may be equivalently parameterized in terms of the degree distribution together with the triangle count.  This is an intuitive property in the context of applications such as social networks, which frequently exhibit both skewed degree distributions and local clustering \citep{davis:asr:1970,holland.leinhardt:ajs:1972,snijders:sn:1981}.  While the homogeneous Markov graphs are prone to degenerate behavior that makes them impractical in most empirical settings \citep{schweinberger:jasa:2011}, they are often useful as building blocks for other model families \citep[see e.g.][]{schweinberger.handcock:jrssB:2015} and they continue to be important objects of theoretical study.  Unfortunately, the dependency properties of the Markov graphs make direct simulation infeasible, and existing applications rely upon MCMC methods for approximate sampling from this family.  Here, we apply our method to the problem of exact sampling from the Markov graphs.  Although we will limit ourselves to the undirected case, the approach generalizes fairly straightforwardly to the directed case as well.

\subsection{Bounds on the Markov Graph Change Scores \label{sec_mgbounds}}

To apply our sampling scheme to the Markov graphs, we must derive the change score bounds $\Delta^L$ and $\Delta^U$.  Since the Markov graph statistics are all subgraph census statistics, these bounds can be expressed directly in terms of the current states of $L$ and $U$ as follows:
\begin{gather}
\Delta^L_{jk}\left(\mathcal{B}^{(i)},\theta\right)_l = \begin{cases} \Delta_{jk}\left(U^{(i)}\right)_l& \theta_l \le 0\\ \Delta_{jk}\left(L^{(i)}\right)_l& \theta_l > 0\end{cases} \label{eq_deltalmg}\\
\Delta^U_{jk}\left(\mathcal{B}^{(i)},\theta\right)_l = \begin{cases} \Delta_{jk}\left(L^{(i)}\right)_l& \theta_l \le 0\\ \Delta_{jk}\left(U^{(i)}\right)_l& \theta_l > 0\end{cases}. \label{eq_deltaumg}
\end{gather}
It then remains to compute $\Delta$.  Change scores for the Markov graph statistics are well-known in the network field (and implemented in packages such as \texttt{ergm} \citep{hunter.et.al:jss:2008}), but for completeness we review them here.  

\subsubsection{$k$-stars}

As noted above, the $k$-star statistic of graph $G$ is the number of copies of $K_{1,k}$ within $G$.  $k$ may take any value from 1 to $n-1$, with the former being simply the edges of $G$ and the latter $G$'s spanning stars.  Let $d_i(y)=\sum_{j=1}^n y_{ij}$ be the \emph{degree} of the $i$th vertex in $G$; then the number of $k$-stars in $G$ is equal to $t_k(y) = \sum_{i=1}^{n} \tbinom{d_i(y)}{k}$.  It follows, then, that the change score for the $k$th star statistic associated with the $\{i,j\}$ edge must be
\begin{align}
\Delta_{ij}(y)_k &= t_k\left(y^+_{ij}\right)-t_k\left(y^-_{ij}\right)\\
&= \binom{d_i(y^+_{ij})}{k}-\binom{d_i(y^-_{ij})}{k} + \binom{d_j(y^+_{ij})}{k}-\binom{d_j(y^-_{ij})}{k}\\
&= \binom{d_i(y^-_{ij})+1}{k}-\binom{d_i(y^-_{ij})}{k} + \binom{d_j(y^-_{ij})+1}{k}-\binom{d_j(y^-_{ij})}{k}\\
&=\binom{d_i(y^-_{ij})}{k-1}+\binom{d_j(y^-_{ij})}{k-1}.
\end{align}
The $k$-star change scores are hence simple functions of the (perturbed) degree distribution.  In practice, this is generally implemented by tracking degrees over time, which avoids the cost of computing $d$ at each update.  In such implementations, $\Delta$ for a $k$-star statistic can be calculated in constant time.

\subsubsection{Triangles}

The triangle statistic of graph $G$ is the number of copies of $K_3$ within $G$.  For an $\{i,j\}$ edge, the number of triangles potentially contributed is equal to the number of two-paths from $i$ to $j$ -- that is, the number of vertices $k \neq i,j$ such that $\{i,k\},\{k,j\}$ are in $G$.  If the $\{i,j\}$ edge is not present, none of these triangles exist; if the edge is present, all of them do.  It follows, then, that the change score for the triangle statistic is simply $\Delta_{ij}(y) = \sum_{k\neq i,j} y_{ik}y_{kj}$.  Although this is a linear-time update, it can be improved in practice via a sparse-matrix implementation which searches $i$ and $j$ for common neighbors.  Average running time in this case is dominated by the mean degree of $G$, which is often much smaller than $n$ for large networks.

\subsection{Numerical Example}

To illustrate the application of our CFTP algorithm to the Markov graphs, we show the results of a simple simulation study using two well-known two-parameter subfamilies.  The first such family is the ``edge clustering'' or two-star model, which consists of the subfamily of Markov graphs parameterized by the one-star (i.e., edge count) and two-star statistics.  Since the first two $k$-star statistics jointly characterize the mean and variance of the degree distribution, this family can be correctly described as the maximum entropy graph distribution obtained by fixing the first two moments of the degree distribution (and nothing else).  Alternately, it can also be thought of as a model in which edges may have a propensity to ``cluster'' around the same vertices (rather than to be scattered at random throughout the graph).  The second subfamily treated here is the ``triangle'' model, which is composed of the Markov graphs parameterized by the one-star and triangle ($K_3$) statistics.  This is arguably the simplest model of structural clustering (in the sense of completed two-paths), a property known to be frequent in social networks since at least the meta-analyses of \citet{davis:asr:1970} and \citet{holland.leinhardt:ajs:1972}.  While neither of these models is typically plausible from a substantive standpoint -- both are excessively homogeneous and prone to degeneracy -- they have played an important theoretical role as ``toy'' models for the exploration of edge dependence \citep[see, e.g.][]{jonasson:jap:1999,haggstrom.jonasson:jap:1999,handcock:ch:2003,burda.et.al:prevE:2004,park.newman:prevE:2004}.  Their use here also allows for comparison with other studies, e.g. \citet{handcock:ch:2003}.

For our illustrative simulation, we employed Algorithm~\ref{alg_erg} to take draws from each of the two submodel families.  In both cases, $\theta_1$ (edge count) was varied from -7.5 to 7.5, with the second free parameter ($\theta_2$ for two-stars, $\theta_3$ for triangles) varied from -0.5 to 0.75; this was done evenly in a 51 by 51 grid of parameter values.  To facilitate comparison with \citet{handcock:ch:2003} (whose study examined the two-star model using complete enumeration), $n$ was fixed to 7 throughout.  In each cell, 500 draws were taken from each of the corresponding models via the perfect sampling algorithm, and various descriptives were computed.  The results of these simulations are summarized in Figures~\ref{f_degandtime}--\ref{f_2stri_trimod}.  

Figure~\ref{f_degandtime} provides an indicator of degeneracy properties of each model (left panels), as well as information on algorithm performance (specifically, the log of the average number of iterations required for coalescence detection).  To assess degeneracy, we examined the total probability of drawing a complete or empty graph ($K_7$ or $N_7$) as a function of model parameters -- although models can exhibit other forms of degeneracy, collapse of the probability distribution into a mixture of complete and empty graphs is a phenomenon of particular interest for these models.  The top left panel of Figure~\ref{f_degandtime} shows the characteristic ``wedge'' pattern of degeneracy identified by \citet{handcock:ch:2003}, with parameters outside of a linearly bounded, triangular region leading to degenerate or near-degenerate mixtures of complete and empty graphs.   A very similar pattern is obtained under the triangle model (lower left panel), although the ``wedge'' is steeper relative to $\theta_1$ than in the two-star model.  This change reflects differences in the trade-off between density and the count of triangles and two-stars (respectively): because the triangle count can vary more readily at constant density, we see greater sensitivity (in terms of convergence to complete or empty graphs) to $\theta_2$ than to equivalent changes in $\theta_3$.  Nevertheless, both subfamilies lead to qualitatively similar regions of non-degeneracy over this portion of the parameter space, and neither is especially well-behaved in this regard.

\begin{figure}
\begin{center}
\includegraphics[width=6in,height=6in]{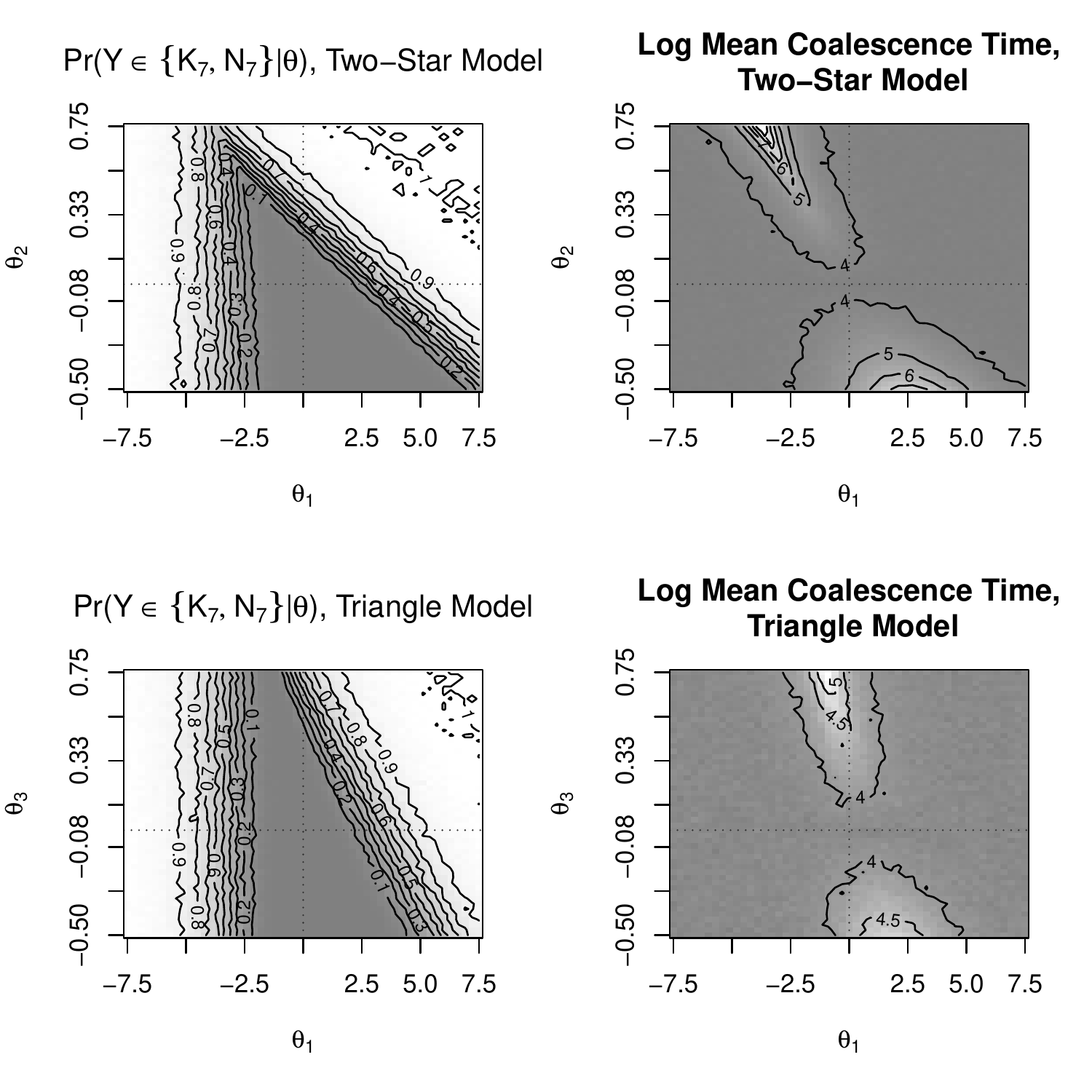}
\caption{\label{f_degandtime} Probability of Drawing Null or Complete Graphs (left) and Log Mean Coalescence Time (right), Two-Star and Triangle Models}
\end{center}
\end{figure}

While one might intuitively suppose that degenerate models would lead to performance problems for the sampling algorithm, this is not necessarily the case.  Although samples generated by our scheme are guaranteed to be from the equilibrium distribution of the model (subject to the usual caveats of pseudo-random number generation), the time needed to generate those samples is dependent upon the mixing properties of the underlying Gibbs sampler: where the sampler mixes poorly, time to coalescence may be extremely long.  In this respect, the right-hand panels of Figure~\ref{f_degandtime} provide news of a mostly salubrious nature.  Mean coalescence time for both the two-star and triangle models is short for the bulk of the parameter space, including most of the degenerate region.  The short coalescence time in the latter case is governed by the uniformity of attraction towards a single degenerate state; most chains for degenerate models quickly collapse into either the complete or empty graph, a process which does not impede performance.  Figures~\ref{f_2stri_2smod} and \ref{f_2stri_trimod} illustrate this dichotomy via the means and standard deviations of two-star and triangle counts for each model.  As the figures show, the bulk of the degenerate region for each model is composed of draws with no two-stars or triangles (in practice, mostly empty graphs), or draws with the maximum number of two-stars and triangles (complete graphs).  Outside of a narrow region of complete/empty mixtures (to which we will turn presently), most models outside the ``wedge'' lead to pure complete or empty distributions and are easy to simulate.

\begin{figure}
\begin{center}
\rotatebox{0}{\resizebox{6in}{6in}{\includegraphics{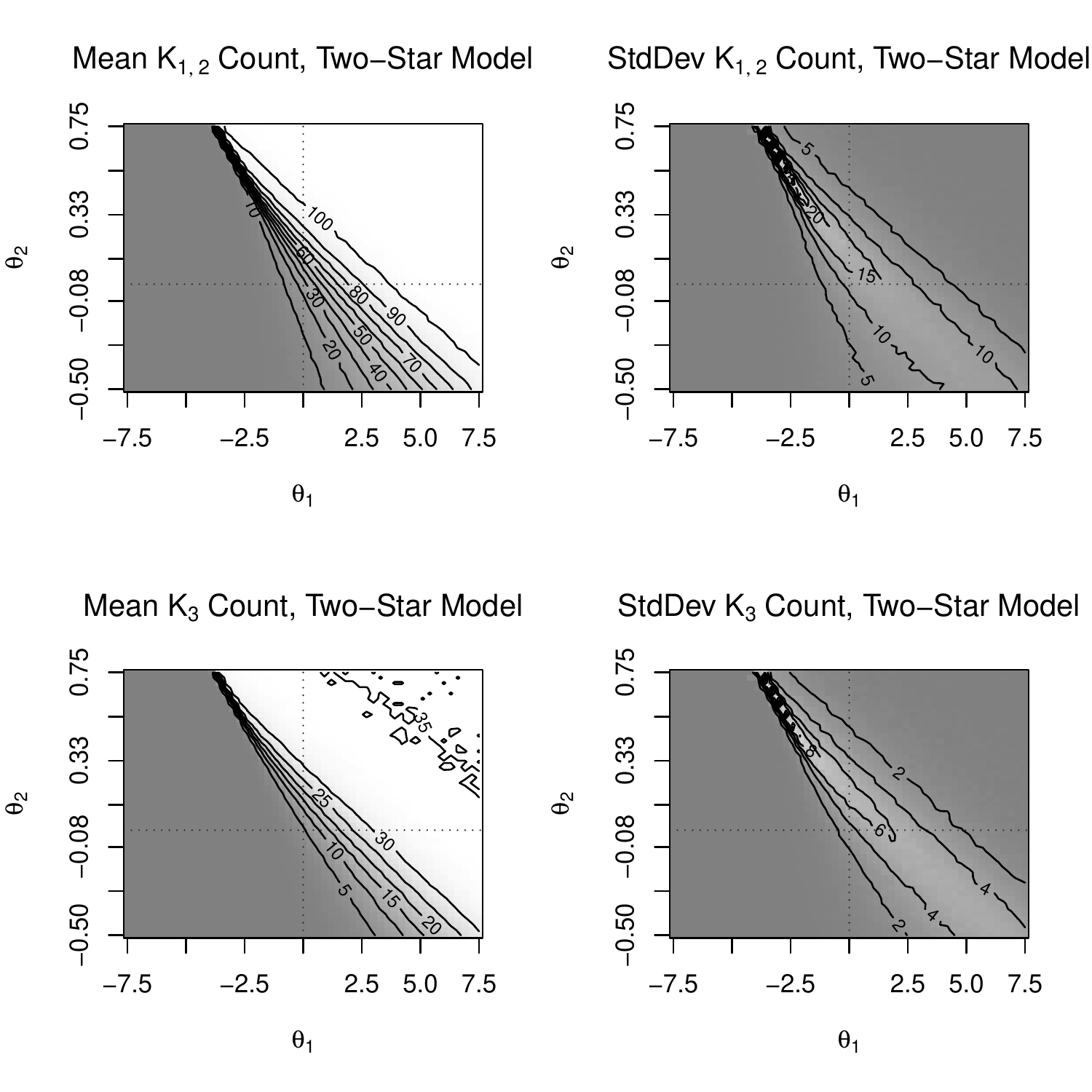}}}
\caption{\label{f_2stri_2smod} Two-Star and Triangle Statistics, Two-Star Model}
\end{center}
\end{figure}

\begin{figure}
\begin{center}
\rotatebox{0}{\resizebox{6in}{6in}{\includegraphics{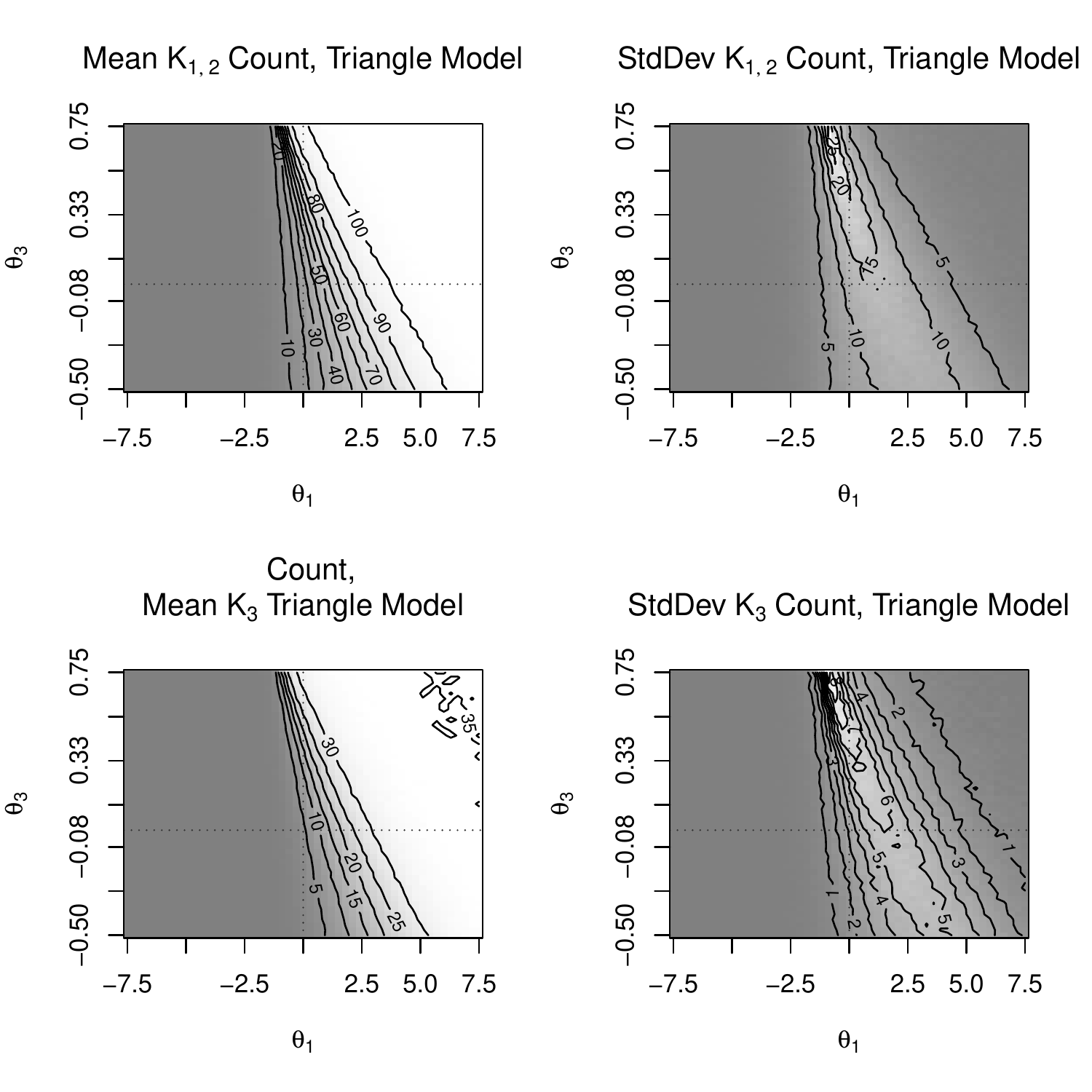}}}
\caption{\label{f_2stri_trimod} Two-Star and Triangle Statistics, Triangle Model}
\end{center}
\end{figure}

This happy state of affairs breaks down, unfortunately, when the forces encouraging high/low density and subgraph formation/dissolution are both extreme and in balance with one another.  This is visible within the right-hand panels of Figure~\ref{f_degandtime} in the longer log-convergence times found for models in the upper central-left and lower central-right regions of the parameter space for each subfamily.  The former regions are in the conventionally degenerate portion of the graph distributions, and correspond to mixtures of complete and empty graphs.  (This can be confirmed by examining the right-hand panels of Figures~\ref{f_2stri_2smod} and \ref{f_2stri_trimod}.)  For these models, $Y$'s transition time between extreme states may be extremely long, and coalescence difficult to obtain.  A related (if more subtle issue) is responsible for the increased coalescence time in the lower central-right regions of the parameter space.  While these models do not readily produce complete or empty graphs, they are near-degenerate in other respects: specifically, they tend to ``crystallize'' into a very small number of isomorphism classes with minimal numbers of two-stars or triangles (respectively).  Where these ``frozen'' structures differ by more than a single edge change, transition times between them may be long, thereby impairing mixing in the same manner as joint convergence to complete/empty graph mixtures.  This behavior was noted by \citet{handcock:ch:2003}, who found elongated, ray-like structures of such non-trivial degeneracy in the two-star model (see also \citet{robins.et.al:ajs:2005}).  Cross-referencing the top right panel of Figure~\ref{f_degandtime} with his results confirms that the longer coalescence times in the central-right region corresponds to the approach of this portion of the parameter space; a similar pattern is observed for the triangle subfamily (bottom right panel).  Taken together, then, our simulations suggest that the CFTP procedure will work well in fully non-degenerate regions of the parameter space (where mixing is not a problem), and in trivially degenerate regions of the parameter space which are characterized by convergence to a single structure.  For models with non-trivial degeneracy (characterized by concentration of probability mass on a small number of structures separated by large numbers of edge changes), coalescence times may become very long.

\section{Extensions \label{sec_ext}}

Although we have focused on the case of ERGs on fixed-order simple graphs, the approach developed here is easily extended to cover other cases.  Here, we briefly describe two of the most obvious: graphs which are directed, have loops, or are edge restricted; and models specified in biased net form.

\subsection{Directed Graphs, Loops, and Edge Restrictions \label{sec_ext_dig}}

Probably the most important extension of the simple procedure discussed here is to the case of graphs which are directed and/or which have loops (i.e., self-ties).  To accommodate the former case, we conduct all updates on the ordered pairs $(i,j)$, instead of the unordered pairs $\{i,j\}$, and relax the assumption of symmetry for $Y$, $L$, and $U$.  To allow loops, we similarly extend the set of possible edge updates to include the multisets $\{i,i\}$ in the undirected or pairs $(i,i)$ in the directed cases (respectively).  Otherwise, no changes are necessary: provided that $t$ is defined appropriately, the same procedure suffices to draw perfect samples from the corresponding graph or digraph distribution.  (Note that the directed case can lead to very different choices of statistics -- and $\Delta$ computation -- in practice.  This is a model parameterization issue, however, rather than a sampling issue.)    

Another common problem is the simulation of ERG models whose support is restricted to some subset $\mathcal{Y}'_n \subset \mathcal{Y}_n$.  This incorporates a large number of special cases, not all of which can be handled using the Gibbs sampler (and hence are capable of being simulated via our procedure).  Although we will not attempt a general treatment of this problem here, one important family of cases is especially easy to accommodate: specifically, ERGs for which particular edges are restricted to be present or absent \emph{ex ante}.  This includes the case of ERGs on bipartite graphs (i.e., graphs such that $V=V_r \cup V_c$, with $E\subseteq V_r \times V_c$), as well as ERGs for egocentric networks (i.e., graphs conditioned to have some known vertex, $v$, as a spanning star).  Our basic procedure may be modified for such models as follows.  Let $\mathcal{Y}'_n$ be the support of the edge-restricted model, and define $K'_n, N'_n$ to be the upper and lower bounds on $\mathcal{Y}'_n$ under the subgraph operator.  Both bounds exist and are unique: they are obtained by treating all free edges as present in the former case, and absent in the latter (leaving restricted edges untouched).  Now, let $\mathcal{E}$ be the set of unrestricted edge variables; we may then sample from $\mathcal{Y}'_n$ by initializing $L=N'_n$ and $U=K'_n$, and choosing updates at random from $\mathcal{E}$.  Since all elements of $\mathcal{Y}'_n$ are reachable through single-edge changes, and since the ordinal properties of the subgraph operator are unchanged, the results described here generalize directly to the restricted case.  (As with the directed case, however, edge restrictions may affect one's choice of $t$.)

\subsection{Biased Net Models \label{sec_bn}} 

One historically important alternative to the use of discrete exponential families to parameterize models for networks with complex dependence has been the ``biased net'' family of stochastic processes introduced by \citet{rapoport:bmb:1949a,rapoport:bmb:1949b,rapoport:bmb:1950}.  Treatment of this family in the literature has not always been consistent; the most inferentially well-developed framework is that discussed by \citet{skvoretz.et.al:sn:2004}, which parameterizes the family via approximations to the full conditionals of each edge.  Although the resulting expressions can become quite complex (depending on the order of the approximation involved), their simplest approximation can be parsimoniously described in terms of a linear form for the conditional log-probability of a non-edge, i.e.
\begin{equation}
\ln \Pr\left(Y_{ij}=0|y^c_{ij},\theta,t\right) \approx t\left(i,j,y^c_{ij}\right)^T \log \theta,
\end{equation}
where $t\in \{0,1,\dots\}^p$ is a vector of sufficient statistics and $\theta \in [0,1]^p$ is a parameter vector.  (Conventionally, the support of the model is taken to be the order-$n$ digraphs, although generalizations to undirected or edge-constrained graphs are straightforward.)  As generally understood, the elements of $t$ refer to counts of edge-formation (or so-called ``bias'') events for the $(i,j)$ edge variable, with $\theta$ being the corresponding conditional probabilities that an edge does \emph{not} form given a particular event.  This interpretation leads to the more conventional form
\begin{equation}
\Pr\left(Y_{ij}=1|y^c_{ij},\theta,t\right) \approx 1-\prod_{k=1}^p\left(1-\theta^*_k\right)^{t_k\left(i,j,y^-_{ij}\right)} \label{e_bn}
\end{equation}
with $\theta^*=1-\theta$.  As this expression suggests, each potential edge within the biased net model can be thought of as being exposed to a series of bias events, each of which independently leads to the formation of the edge with probability $\theta^*_k$.  If no bias event triggers edge formation, then the edge is taken to be absent.  Generally, $t$ is taken to include a constant event (representing a base rate of tie formation), with events relating to reciprocated edges and shared partners being the most widely researched (see, e.g., \citet{skvoretz:sn:1985,skvoretz:sn:1990,skvoretz.et.al:sn:2004}).

While the above is an extremely attractive and intuitive framework, it suffers from the problem that an expression for the joint distribution of $Y$ under such a family is not known; indeed, it can be shown that some such families have \emph{no} joint distribution (i.e., there exists no random graph $Y$ whose full conditionals are compatible with Equation~\ref{e_bn} for some choices of $t$ and $\theta$).  The first-order biased net model family is as such ill-posed.\footnote{\citet{skvoretz.et.al:sn:2004} also consider higher order approximations, which may not suffer similar difficulties.}  However, a very similar family can be constructed, which preserves the intuition of the original.  Specfically, let us imagine a social process in which $Y$ evolves in discrete steps $\ldots,Y^{(0)},Y^{(1)},\ldots$ such that at iteration a single, randomly chosen $i,j$ edge is either added to or removed from the graph.  Given that $(i,j)$ refers to the edge selected at arbitrary time $t$, the graph then evolves via the following process:
\begin{equation}
Y^{(t+1)}_{gh} = \begin{cases} 1 & (g,h)=(i,j) \ \mathrm{and}\ u^{(t)} < 1-\prod_{k=1}^p\left(1-\theta^*_k\right)^{t_k\left(i,j,\left(Y^{(t)}\right)^-_{ij}\right)} \\ 0 & (g,h)=(i,j) \ \mathrm{and}\ u^{(t)} \ge 1-\prod_{k=1}^p\left(1-\theta^*_k\right)^{t_k\left(i,j,\left(Y^{(t)}\right)^-_{ij}\right)} \\ Y_{gh}^{(t)} & (g,h) \neq (i,j)\end{cases} \label{e_pgibbs}
\end{equation}
where, as previously, $\ldots,u^{(0)},u^{(1)},\ldots$ is a set of iid uniform deviates on the $[0,1]$ interval.  This process obviously defines a Markov chain that closely resembles a random-update Gibbs sampler, and indeed such a process is a Gibbs sampler for the joint distribution associated with the full conditionals of Equation~\ref{e_bn} where such a distribution exists.  Where such a distribution does not exist, however, the associated Markov chain is still well-defined, and can be viewed as a model of social process in its own right.\footnote{Such chains (based on approximate full conditionals with no well-formed joint distribution) are sometimes called \emph{pseudo-Gibbs samplers} \citep{chen.ip:jscs:2014}.}  Substantively, one interpretation of such a model is as follows.  At each time step, a randomly selected individual considers the state of his or her relationship with another randomly selected individual.  The current state of the network generates a set of bias events, any of which may trigger the creation or retention of an edge.  If no such ``trigger'' resolves, the relationship relaxes to (or stays at) the null state.  A new pair is then considered, and the process continues in like vein.

\citet{butts:sw:2000} implements approximate simulation of draws from the above process using MCMC.  The method used is sequential sampling using Equation~\ref{e_pgibbs}; under standard regularity assumptions, a sample $y^{(1)},y^{(2)},\ldots$ from this process converges in the limit of iterations to its (uncharacterized) equilibrium distribution.  Although straightforward, this method is clearly approximate in the finite-sample case.  By applying a variation on the simulation scheme presented in this paper, however, it is possible to obtain exact samples from the majority of biased net models employed in the literature (as expressed in Markov chain form).  Specifically, we consider here the case in which, for all pairs $(i,j)$ and all statistics $t_k$, $t_k(i,j,y)\le t_k(i,j,y')$ for all $y \subseteq y'$. This monotonicity condition is satisfied by the well-known ``parent,'' ``sibling,'' and ``double-role'' biases employed in the biased net literature (as well as variants thereof).

Our approach is straightforward.  As usual, we initialize $L$ to the empty graph and $U$ to the complete graph (in their directed guises) at some time $-i$.  We then evolve $L$ and $U$ toward time 0 by the rule of Equation~\ref{e_bn}, using a shared sequence of random edge variables to update and random ``coins'' $u^{(-j)} \sim \mathrm{Unif}(0,1)$.  Since $\theta, t$ are non-negative and $t$ is weakly monotone with respect to the subgraph operator, it follows that the conditional probability that $Y^{(-j)}_{kl}=1$ given $Y^{(-j-1)}$ is also weakly monotone with respect to subgraph ordering.  Thus $\Pr\left(L^{(-j)}_{kl}=1|L^{(-j-1)},\theta,t\right)\le\Pr\left(Y^{(-j)}_{kl}=1|Y^{(-j-1)},\theta,t\right)\le\Pr\left(U^{(-j)}_{kl}=1|U^{(-j-1)},\theta,t\right)$ for each $j<i$, and $Y$ is bounded by $L$ and $U$.  By the usual arguments, then, $L^{(-j)}=U^{(-j)}$ is a sufficient condition for coalescence, and the corresponding value of $Y^{(0)}$ is a draw from the equilibrium distribution of $Y$.  Since $Y$ was constructed to be a Gibbs sampler for the biased net model with statistics $t$ and parameter vector $\theta$, it follows that the procedure produces samples from the target model.

In passing, it should be noted that little is known regarding exponential family representations of biased net processes (beyond trivial cases such as the parent bias model).  A general mapping from biased net to ERG parameterizations (and, where possible, the reverse) would serve to make this venerable line of work more accessible to researchers within the broader network statistics community.

\subsubsection{Example: Probing the Sibling Bias}

In the language of biased net theory, a ``sibling'' bias event for the $(i,j)$ edge variable is produced by every vertex $k$ such that $k \to i$ and $k \to j$ (i.e., $i$ and $j$ have an incoming shared partner).  Formally, the sibling bias is parameterized via the statistic $t(i,j,y)=\sum_{k=1}^n y_{ki}y_{kj}$.  Clearly, sibling events promote transitivity in the sense that they enhance the conditional probability of observing $i,j,k$ triples such that $k \to i$, $i\to j$, and $k \to j$, though they operate by encouraging the formation of ``pre-closed'' two-paths rather than by encouraging the closure of open two-paths.  Given their distinct mode of operation, it is natural to ask whether the sibling bias leads to phase transitions analogous to those of the edge-triangle model \citep{strauss:siam:1986}.  Using the exact sampling mechanism described above, we can answer this question through simulation.

The top left panel of Figure~\ref{f_sibsim} shows the distributions of density and transitivity for exact draws from a 25-node network biased net model with a baseline edge probability ($d$) of 0.125 and sibling effects ($\sigma$) ranging from 0 to 1 in increments of 0.02; 500 draws were taken per parameter value.  As can be seen, density and transitivity track closely over the entire parameter space, with a very sharp transition from a relatively sparse, intransitive regime below $\sigma \approx 0.1$ to an extremely dense, transitive regime.  This transition is very similar to the phenomenon observed in the edge-triangle ERGM.  In the ERGM case, another standard observation is that the transition to the dense phase occurs with increasing graph order, for a fixed triangle parameter.  The top right panel of Figure~\ref{f_sibsim} examines the parallel question for the sibling bias, here fixing the parameter at 0.1 (with the baseline tie probability set to a constant mean degree of 3) and varying the number of vertices from 5 to 60 (5000 draws per condition).  While we see considerable variability in the case of very small graphs, increasing $|V|$ leads the system to settle into a sparse phase before transitioning sharply to a nearly complete phase at $|V|\approx 25$.  Thus the baseline/sibling bias model strongly resembles the edge-triangle model in behavior, despite being very differently parameterized.

\begin{figure}
\begin{center}
\includegraphics[width=6in,height=4.8in]{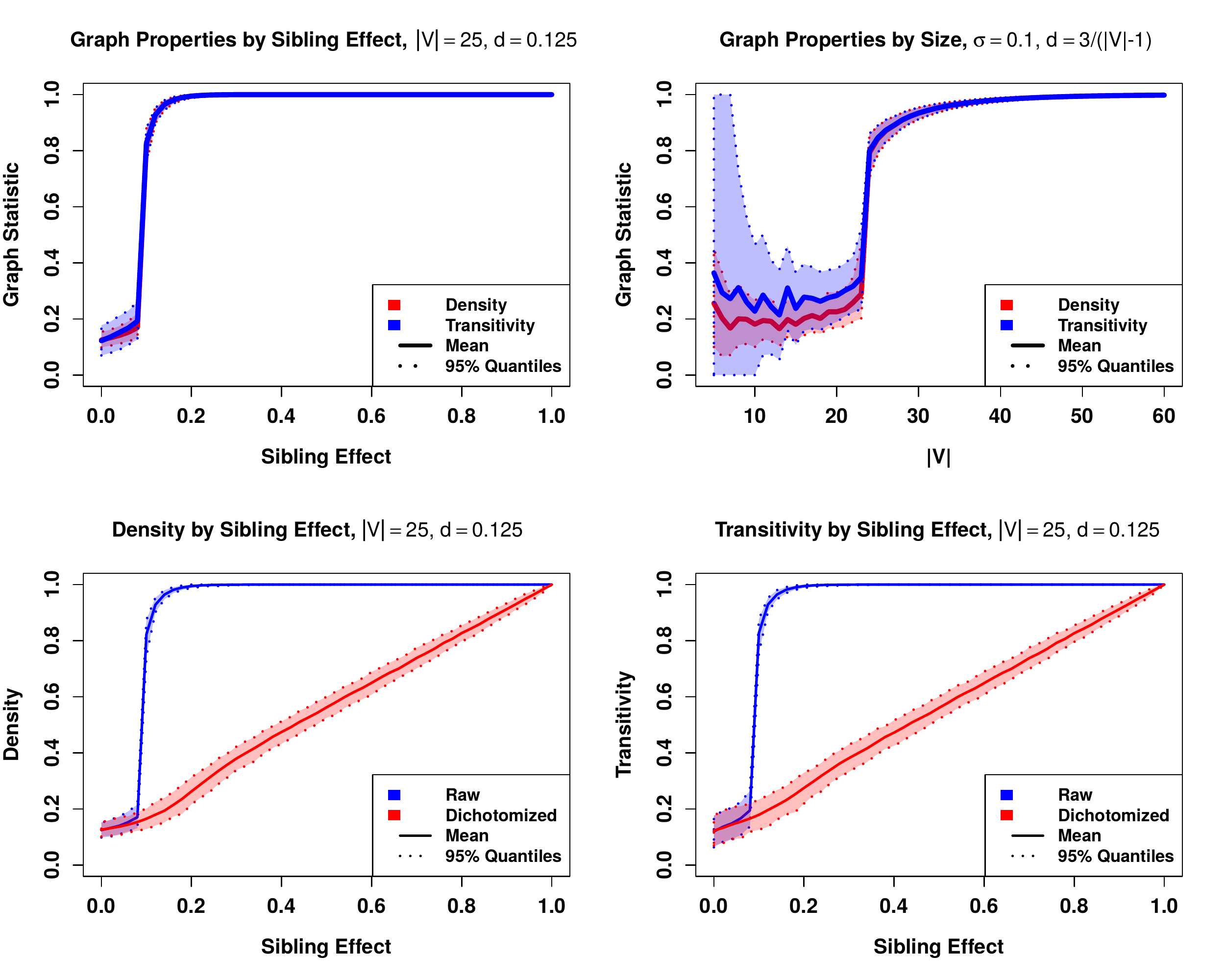}
\caption{\label{f_sibsim} Exact Biased Net Simulations.  (top left) Density and Transitivity by Sibling Effect Strength.  (top right)  Density and Transitivity by Size.  (bottom left) Density by Sibling Effect Strength, Raw versus Dichotomized Statistics. (bottom right) Transitivity by Sibling Effect Strength, Raw versus Dichotomized Statistics. }
\end{center}
\end{figure}

How could this behavior be altered?  One idea is to change the nature of the sibling bias.  Conventionally, we assume that every incoming shared partner creates an independent opportunity for a tie to form, thus leading to a cascade of runaway edge formation once a critical density threshold is reached.  An alternative assumption is that only the first incoming shared partner is important for prompting tie formation: if a tie is not induced by this bias event, subsequent shared partners have no additional effect.  We can implement this via a dichotomized version of the incoming shared partner statistic shown above, with the statistic being equal to 1 if any shared partner is present, and 0 otherwise.  Because this statistic is still monotone, we can employ it with our exact sampler.  Results from simulations varying $\sigma$ in dichotomized and ``raw'' form ($|V|=25$, $d=0.125$, 500 draws per condition) are shown in the bottom panels of Figure~\ref{f_sibsim}.  As can be seen, the dichotomized sibling effect successfully eliminates the phase transition behavior arising from the convetional sibling bias, instead leading to density and transitivity values that scale almost linearly as a function of $\sigma$.  

These simple examples illustrate how the exact sampling method can be used to explore the behavior of biased net models, even in regimes for which those models are not well-behaved.  The simulations reveal that simple examples of this model class poses many of the same challenges found in simple ERGMs, but also that some of the same strategies used in ERGMs to avoid degeneracy (here, bounding the strength of a dependence effect) can be adapted to biased nets.  Having a tool for exact simulation from this model class greatly facilitates exploration of alternative options for model parameterization, and will hopefully encourage more work in this area.

\section{Conclusion}

In this paper, we have introduced a method for drawing ``perfect'' samples from random graph distributions in ERG form, as well as certain other model families (including one implementation of the classic ``biased net'' framework).  This method uses a variant of Propp and Wilson's (\citeyear{propp.wilson:rsa:1996}) Coupling From The Past, with the single edge update Gibbs sampler used as the underlying chain.  Although this chain is not monotone, coalescence detection is possible by ``sandwiching'' the states of the Gibbs sampler between the states of two dominating processes, all of which are guaranteed to satisfy a partial order condition (namely, subgraph inclusion).  For ERGs based on the most common types of statistics (subgraph census statistics), computation is fairly straightforward, and requires only change scores on the two bounding processes (resulting in updates which are of the same complexity as the conventional Gibbs sampler).  We have illustrated the use of this method via an application to the Markov graphs of \citet{frank.strauss:jasa:1986}, exploring its behavior for two model subfamilies (the two-star or edge clustering and triangle models, respectively).  Algorithm performance is good for nondegenerate models, and for degenerate models in which the model distribution collapses onto a single structure; degenerate models which collapse onto mixtures of structures which cannot be reached via single edge changes lead to very poor mixing in the underlying Gibbs sampler, and thus long coalescence times.  The method shown here is not therefore a panacea for model degeneracy (a well-known challenge with ERGs), although it does provide an assurance that samples obtained are exact (to the limit of one's underlying numerical and pseudo-random infrastructure).  Arguably, another virtue of this method versus conventional Markov chain Monte Carlo is that poor performance in the case of non-trivial degeneracy is made immediately evident by long simulation times, rather than being concealed in the sequence of (potentially difficult to diagnose) graph statistics.  This may allow for faster identification of pathological cases, and reduced risk of erroneous generalization from inadequate MCMC samples.  Given the growing importance of random graph models throughout the social and biological sciences, it is hoped that approaches such as this one will facilitate the study of relational data across a range of substantive applications.

\bibliography{ctb}


\end{document}